\documentclass[letterpaper,11pt]{article}
\usepackage[margin=1in]{geometry}
\usepackage{setspace,graphicx,url,hyperref,natbib}
\usepackage{amsmath,amsthm,amssymb,amsfonts}
\usepackage[linesnumbered,ruled]{algorithm2e}
\newcommand{\N}{\ensuremath{\mathbb{N}}}
\newcommand{\R}{\ensuremath{\mathbb{R}}}

\newtheorem{theorem}{Theorem}
\newtheorem{proposition}{Proposition}

\title{Exact computation of the network modularity with branch-and-bound}
\author{Emily Weng and Georg Hahn}
\date{Harvard T.H.\ Chan School of Public Health, Boston, MA 02115, USA}

\begin{document}
\maketitle

\abstract{We consider the computation of the network modularity, a common measure of the structure of networks or graphs that quantifies the division of a network into modules or clusters. The computation of the network modularity is an NP-hard problem, thus making its exact computation often infeasible in practice. This has resulted in the development of several heuristics in the literature. In this work, we consider the exact computation of the network modularity with the help of a branch-and-bound algorithm. Our algorithm is guaranteed to find the optimal clustering that maximizes the network modularity, however it achieves this without a full exploration of the search space. This is accomplished by bounding partial solutions and discarding them if it can be foreseen that a partial solution will not yield an improvement over an existing solution. We assess our algorithm with respect to accuracy and runtime and compare it to a brute-force approach as well as two benchmark algorithms, simulated annealing and the Louvain method.}

\onehalfspacing
\section{Introduction}
\label{sec:introduction}
We consider the computation of the network modularity, a measure that indicates the degree of separation of a graph or network into modules or clusters. We are given a graph or network $G=(V,E)$ with $n \in \N$ vertices $V=\{1,\ldots,n\}$ and edge set $E \subseteq V \times V$. We define the node degree $k_v$ for each $v \in V$ and the number of edges $m = |E|$. The network modularity is defined as
\begin{align}
    Q = \frac{1}{2m} \sum_{vw} \sum_{r=1}^R \left[ A_{vw} - \frac{k_v k_w}{2m} \right] S_{vr} S_{wr} = \frac{1}{2m} \text{Tr}(S^\top B S),
    \label{eq:modularity}
\end{align}
where $A \in \R^{n \times n}$ is the adjacency matrix of $G$, $R \in \N$ is the number of clusters (assumed to be known in advance), $S_{vr} \in \{0,1\}$ is an indicator encoding that vertex $v$ is assigned to cluster $r \in \{1,\ldots,R\}$, $S \in \R^{n \times R}$ is the matrix with entries $S_{vr}$ for $v \in V$ and $r \in \{1,\ldots,R\}$, $B \in \R^{n \times n}$ is a matrix with entries $B_{vw}=A_{vw}-\frac{k_v k_w}{2m}$ for all $v,w \in V$, and Tr denotes the trace of a matrix. Each vertex $v$ can only be assigned to one cluster. The modularity value ranges from $-1/2$ (non-modular clustering) to $1$ (fully modular clustering). The task of modularity optimization consists in the allocation of each vertex $v \in V$ to exactly one cluster $r \in \{1,\ldots,R\}$ (encoded with $S_{vr}=1$) such that $Q$ is maximized. As a problem belonging to the class of binary quadratic optimization (with constraints on the assignments of the variables $S_{vr}$ for $v \in V$ and $r \in \{1,\ldots,R\}$), the global maximization of eq.~\eqref{eq:modularity} is NP-hard \citep{Boroujeni2022}.

As the computation of the network modularity is an NP-hard problem, its exact computation is often infeasible, thus allowing for a rich variety of approximation approaches available in the literature. Known algorithms can be grouped into roughly seven categories according to their mode of action (divisive hierarchical methods, direct optimization, agglomerative methods, multilevel algorithms, dynamics-based methods, statistical inference and generative models, and graph neural network methods). A detailed overview of current methodology is given in Section~\ref{sec:literature}.

In this contribution, we focus on the exact computation of the network modularity value, that is the task of finding an assignment of $S_{vr} \in \{0,1\}$ for each vertex $v \in V$ that globally maximizes eq.~\eqref{eq:modularity}. We attempt this with the help of a branch-and-bound algorithm, meaning an algorithm that explores the entire search space (of size $O(R^{|V|})$) in a recursive fashion while aiming to omit the exploration of partial solutions as soon as it can be foreseen that they will not yield the global optimum. This is achieved with the help of valid (guaranteed) bounds on any partial solution, that is, by computing guaranteed bounds on how much to the above (or below) the modularity can change under any completion of the partial assignment. If, based on these bounds, it is obvious that a partial solution will not yield a new global maximum, we omit its exploration and proceed with the next unexplored partial solution.

We assess our proposed method with respect to the modularity-value gap and runtime in seconds in three experiments on simulated data. The modularity-value gap is defined as the difference between the maximum value of eq.~\eqref{eq:modularity} achieved via brute-force enumeration and the modularity value returned by any method under investigation. We compare our algorithm to two classic benchmarks available in the literature, simulated annealing \citep{Kirkpatrick1983} and the Louvain method \citep{Blondel2008}.

This article is structured as follows. A detailed literature review can be found in Section~\ref{sec:literature}. We introduce the proposed bounds on the network modularity and the resulting branch-and-bound algorithm in Section~\ref{sec:methods}. In Section~\ref{sec:results}, we assess our proposed algorithm in simulations and compare it to two benchmark algorithms (simulated annealing and the Louvain method) with respect to accuracy and runtime. The article concludes with a discussion in Section~\ref{sec:discussion}. Throughout the article, the length of a vector $v \in \R^n$ is denoted as $|v|$.

All implementations used in this article are freely available on GitHub \citep{github}.

\subsection{Literature review}
\label{sec:literature}
The optimization of eq.~\eqref{eq:modularity} has prompted the development of a variety of (exact and approximation) algorithms in the literature. First, there are optimized software tools for the exact optimization of eq.~\eqref{eq:modularity}. The most notable representatives are Gurobi \citep{gurobi} and CPLEX \citep{cplexmanual}, which handle general instances of linear, quadratic, and mixed-integer programming. Both tools are able to find the global maximum of eq.~\eqref{eq:modularity}, however they are not free to use. Moreover, they employ general purpose bounds for the exploration of the search space, whereas we employ a new bound specifically tailored for the modularity measure. Importantly, both Gurobi and CPLEX are not able to guarantee a runtime that is faster than exponential, due to the NP-hard nature of the modularity optimization problem. This means that their asymptotic runtime is not faster than the one of a brute-force solver, although Gurobi and CPLEX can be much faster in practical applications.

Second, there exists an extensive literature on approximate (heuristic) algorithms to maximize eq.~\eqref{eq:modularity}. Notably, these approximation algorithms can be categorized into different algorithm classes, which are briefly summarized in the following paragraphs.

One class of modularity maximization algorithms iteratively removes edges from a network to split it into communities. For instance, the Girvan-Newman algorithm \citep{Girvan2002} iteratively removes the edge with the highest betweenness centrality (a measure of centrality in a graph based on shortest paths), yielding nested communities. Similarly, the divisive algorithm of \cite{NewmanGirvan2004} iteratively removes edges from the network to split it into communities while updating betweenness measures. A greedy algorithm which is optimized for speed is also available in the literature \citep{Newman2004}. Another algorithm which iteratively bisects a given network into communities is the spectral algorithm of \cite{Newman2006}, whereby the bisectioning step is based on the leading eigenvector of the modularity matrix.

Multilevel algorithms, with its two most prominent representatives being the Louvain \citep{Blondel2008} and Leiden \citep{Traag2019} methods, perform two-phase iterative modularity maximization. In the first phase, the modularity is greedily increased by assigning nodes to neighbor communities while in the second phase, communities are collapsed into super-nodes before iteratively repeating the process. The Leiden method improves on the Louvain method with a refinement phase that guarantees internally connected, well-separated communities.

Dynamics-based methods make use of a random walk formulation of the network modularity \citep{Rosvall2008}. In \cite{Pons2005}, the authors define node similarity via short random walks and apply agglomerative clustering on the resulting distance matrix. Label propagation methods work by assigning a node to the majority label of its neighbors until convergence \citep{Raghavan2007}.

Statistical inference and generative methods are based on the idea of fitting a stochastic block model to the given network \citep{Holland1983, Karrer2011}, instead of approximating the modularity directly. Stochastic block models are generative models for random graphs that tend to produce graphs containing communities. A nonparametric Bayesian approach can be found in \cite{Peixoto2014}.

Graph neural network methods use so-called deep modularity networks \citep{Tsitsulin2023} or graph autoencoders to learn to detect communities using both structure and node attributes.

\section{Methods}
\label{sec:methods}
This section starts by briefly introducing the general idea of our branch-and-bound algorithm (Section~\ref{sec:idea}). We proceed by stating a bound on the network modularity (Section~\ref{sec:bounds}) which will lead to our proposed branch-and-bound algorithm (Section~\ref{sec:algorithm}).

\subsection{General idea}
\label{sec:idea}
The task is to assign binary indicators $S_{vr} \in \{0,1\}$ for each vertex $v \in V$ such that (1) each vertex is assigned to only one cluster $r \in \{1,\ldots,R\}$ (encoded with $S_{vr}=1$), and (2) the modularity value of eq.~\eqref{eq:modularity} is maximized. We achieve this with an exhaustive exploration of all possible assignments of the vertices in $V$ to the $R$ clusters, meaning an exhaustive assignment of $S_{vr}=1$ for all $v \in V$ and all $r \in \{1,\ldots,R\}$. This approach has exponential effort $O(R^{|V|})$.

In order to accelerate the brute-force enumeration, we consider bounding the modularity value for partial assignments, that is assignments in which only a subset of the vertices $v \in V$ are assigned to a cluster. If in such a case, it can be foreseen that the upper bound on the modularity value of an incomplete assignment is below the best known assignment found so far, the exploration of the incomplete assignment can be omitted.

\subsection{A bound on the modularity value}
\label{sec:bounds}
We consider eq.~\eqref{eq:modularity} and rewrite it using the matrix $B_{vw}=A_{vw}-\frac{k_v k_w}{2m}$ as
\begin{align}
    Q = \frac{1}{2m} \sum_{vw} \sum_{r=1}^R B_{vw} S_{vr} S_{wr} = \frac{1}{2m} \sum_{vw} B_{vw} \delta_{s_v,s_w},
    \label{eq:bounds}
\end{align}
where $m=|E|$ and $A$ is the adjacency matrix of the graph or network $G$ (see Section~\ref{sec:introduction}). The modularity value can also be expressed more compactly using the Kronecker delta function $\delta_{s_v,s_w}$, where $s_v \in \{1,\ldots,R\}$ denotes the cluster that each vertex $v$ is assigned to. The Kronecker delta function has the property that $\delta_{i,j}=1$ if $i=j$, and zero otherwise.

Consider the summand $B_{vw} \delta_{s_v,s_w}$ in eq.~\eqref{eq:bounds}. There are three options depending on whether none, one, or both of the variables $s_v$ or $s_w$ are unset (i.e., set to $NA$). If both $s_v$ and $s_w$ are set to a value (other than $NA$) and $\delta_{s_v,s_w}=1$ is satisfied, we compute the contribution $B_{vw} \delta_{s_v,s_w}$ to the overall sum in $Q$. We also update the lower and upper bounds ($L_Q$ and $U_Q$) on $Q$ to ensure that $L_Q \leq Q \leq U_Q$ holds true at all times. If both $s_v$ and $s_w$ are unset (meaning they are set to $NA$), but $v=w$, then we know $\delta_{s_v,s_w}=1$ for any assignment they will take in the optimal solution, so we are allowed to also add $B_{vw}$ to $Q$ as well as $L_Q$ and $U_Q$. If one or both of $s_v$ or $s_w$ are unset and $v \neq w$, we modify the bounds only: if $B_{vw}<0$, it contributes to the lower bounds only, otherwise it contributes to the upper bound only. A summary of the bound computation can be found in Algorithm~\ref{algo:bounds}. Since Algorithm~\ref{algo:bounds} iterates over all $(v,w) \in V \times V$, its runtime is easily bounded as $O(|V|^2)$.

\begin{algorithm}[t]
\caption{Modularity bounds}
\label{algo:bounds}
    \SetKwInOut{Input}{input}
    \SetKwFor{Loop}{repeat}{}{end}
    \Input{matrix $B$, partial assignment $s_v$ for all $v \in V$ denoting known clusters with $s_v \in \{1,\ldots,R\}$ and $s_v=NA$ otherwise, number of edges $m=|E|$\;}
    $Q \leftarrow 0$\;
	$L_Q \leftarrow 0$\;
	$U_Q \leftarrow 0$\;
	\For{$(v,w) \in V \times V$}{
        \If{$(s_v \neq NA \textnormal{ and } s_w \neq NA \textnormal{ and } s_v=s_w) \textnormal{ or } (s_v=s_w=NA \textnormal{ and } v=w)$}{
            $Q \leftarrow Q + B_{vw}$\;
            $L_Q \leftarrow L_Q + B_{vw}$\;
            $U_Q \leftarrow U_Q + B_{vw}$\;
        }
        \Else{
				\lIf{$B_{vw}<0$}{$L_Q \leftarrow L_Q + B_{vw}$}
				\lElse{$U_Q \leftarrow U_Q + B_{vw}$}
        }
    }
	$Q \leftarrow Q/(2m)$\;
	$L_Q \leftarrow L_Q/(2m)$\;
	$U_Q \leftarrow U_Q/(2m)$\;
    \Return{$Q,L_Q,U_Q$}\;
\end{algorithm}

Let $V_1 = \{ v \in V: s_v \neq NA \}$ and $V_2 = \{ v \in V: s_v=NA \}$ be the sets of indices for which $s_v$ is assigned to a cluster or unassigned, respectively. The two sets satisfy $V=V_1 \dot{\cup} V_2$. Algorithm~\ref{algo:bounds} returns valid lower and upper bounds satisfying $L_Q \leq Q(s) \leq U_Q$ for any partial assignment, meaning for any given labels $s_v \neq NA$ for $v \in V_1$, the bounds are valid for any way the remaining $s_v$, $v \in V_2$, are assigned to the clusters in the set $\{1,\ldots,R\}$. This is formalized in the following proposition.

\begin{proposition}
\label{prop:bounds}
For any partial assignment of $s$, with assigned entries given in $V_1$ and unassigned entries in $V_2$, Algorithm~\ref{algo:bounds} returns $L_Q$ and $U_Q$ satisfying $L_Q \leq Q(s) \leq U_Q$.
\end{proposition}
\begin{proof}
Let $s$ be an arbitrary partial assignment, consisting of $s_v \in \{1,\ldots,R\}$ for $v \in V_1$ and $s_v=NA$ for $v \in V_2$. In eq.~\eqref{eq:bounds}, all pairs with $v,w \in V_1$ contribute $B_{vw} \delta_{s_v,s_w}$ to the modularity value. Therefore, when iterating over all $(v,w) \in V \times V$ in Algorithm~\ref{algo:bounds}, we add their contribution $B_{vw}$ to $Q$, $L_Q$, and $U_Q$. This also applies to the diagonal case $v=w$ when $s_v=s_w=NA$, since we know that $B_{vw}$ will eventually be added to $Q$ irrespectively of the actual value of $s_v$ (and thus $s_w$). In the off-diagonal case $v \neq w$, if either $s_v=NA$ or $s_w=NA$ or both are $NA$, we know from eq.~\eqref{eq:bounds} that $B_{vw} \delta_{s_v,s_w}$ will eventually evaluate to either zero or $B_{vw}$ depending on the final values of $s_v$ and $s_w$. If $B_{vw}<0$, the modularity value will either remain unchanged (if $s_v \neq s_w$) or decrease (if $s_v=s_w$) once these variables are set. If $B_{vw}>0$, the modularity value will either remain unchanged or increase. Therefore, to arrive at valid bounds, we only need to add $B_{vw}$ to $L_Q$ if $B_{vw}<0$. If $B_{vw}>0$, we add it to $U_Q$ only.
\end{proof}

\subsection{The branch-and-bound algorithm}
\label{sec:algorithm}
The proposed branch-and-bound algorithm is divided into two subroutines. The first is Algorithm~\ref{algo:explore} which recursively explores all possible assignments of $s_v \in \{1,\ldots,R\}$ for all $v \in V$. Algorithm~\ref{algo:explore} includes the best solution found so far in the variable $s^\ast$ (here we denote $s^\ast=(s_1^\ast,\ldots,s_n^\ast)$ in vector notation) each time it calls itself. A pointer $p \in \{1,\ldots,n\}$ indicates the current node which is being explored. If $p<n$, meaning not all vertices are assigned to a cluster, the algorithm bounds the modularity value for the current partial assignment in $s=(s_1,\ldots,s_n)$ with the help of Algorithm~\ref{algo:bounds}: if the upper bound $U_Q$ for the partial assignment in $s$ is below $Q(s^\ast)$, the modularity value of the best known assignment $s^\ast$ at any point while the algorithm is running, then $s^\ast$ can be safely returned and the exploration of the partial assignment $s$ can be stopped. Otherwise, the pointer $p$ is increased, and we set $s_p \leftarrow r$ for all possible $r \in \{1,\ldots,R\}$, followed by recursive calls of Algorithm~\ref{algo:explore}. If $p=n$, meaning that all vertices are assigned to a cluster, we compare the modularity value $Q(s)$ to $Q(s^\ast)$. If $Q(s)>Q(s^\ast)$ we return $s$. Otherwise, we return $s^\ast$.

\begin{algorithm}[t]
\caption{Recursive exploration}
\label{algo:explore}
    \SetKwInOut{Input}{input}
    \SetKwFor{Loop}{repeat}{}{end}
    \Input{matrix $B$, number of vertices $n=|V|$, number of edges $m=|E|$, pointer $p \in \N$, partial assignment $s_v$ for all $v \in V$, best solution $s^\ast_v$ for all $v \in V$, number of clusters $R \in \N$\;}
    \If{$p<n$}{
        $\underline{~~},L_Q,U_Q \leftarrow$ Algorithm~\ref{algo:bounds}$(B,s,m)$\;
        \lIf{$U_Q < Q(s^\ast)$}{\Return{$s^\ast$}}
        $p \leftarrow p+1$\;
        \For{$r \in \{1,\ldots,R\}$}{
            $s_p \leftarrow r$\;
            $s^\ast \leftarrow$ Algorithm~\ref{algo:explore}$(B,n,m,p,s,s^\ast,R)$\;
            $s_p \leftarrow NA$\;
        }
        \Return{$s^\ast$}\;
    }
    \Else{
        \lIf{$Q(s)>Q(s^\ast)$}{\Return{$s$}}
        \lElse{\Return{$s^\ast$}}
    }
\end{algorithm}

Our full algorithm (Algorithm~\ref{algo:branch}) follows immediately from Algorithm~\ref{algo:explore}, which is stated separately from the main algorithm in order to allow for clean recursive calls. In Algorithm~\ref{algo:branch}, we first compute the matrix $B$ (see Section~\ref{sec:introduction}) from the adjacency matrix $A$ underlying the network $G$ for which the modularity value is computed, initialize the partial solution $s_v$ to $NA$ for all $v \in V$ and the best known solution $s_v^\ast$ to the same cluster (cluster number one without loss of generality) for all $v \in V$, and call Algorithm~\ref{algo:explore} for recursive exploration with branch-and-bound. Note that the pointer $p$ is initialized with 1, meaning $s_1=1$ stays fixed throughout the exploration with Algorithm~\ref{algo:explore}. This is due to the fact that the cluster labels are interchangeable, and without loss of generality we fix $s_v=1$. The correctness of Algorithm~\ref{algo:branch} is formalized in the following theorem.

\begin{theorem}
\label{thm:correctness}
The modularity value $Q(s^\ast)$ returned by Algorithm~\ref{algo:branch} is the global maximum.
\end{theorem}
\begin{proof}
The optimality of Algorithm~\ref{algo:branch} follows by construction. Algorithm~\ref{algo:branch} recursively explores all possible assignments of $s_v \in \{1,\ldots,R\}$ for $v \in V$ by calling Algorithm~\ref{algo:explore}, which computes the bounds $L_Q$ and $U_Q$ on all partial assignments with the help of Algorithm~\ref{algo:bounds}. These bounds satisfy $L_Q \leq Q(s) \leq U_Q$ for any partial assignment $s$ according to Proposition~\ref{prop:bounds}. Therefore, by pruning partial assignments with $U_Q < Q(s^\ast)$, where $Q(s^\ast)$ is the current best solution found during exploration, Algorithm~\ref{algo:explore} never dismisses the branch containing the optimal solution (the one leading to the global maximum). Therefore, the result of Algorithm~\ref{algo:branch} is optimal.
\end{proof}

\begin{algorithm}[t]
\caption{Modularity value computation via branch-and-bound}
\label{algo:branch}
    \SetKwInOut{Input}{input}
    \SetKwFor{Loop}{repeat}{}{end}
    \Input{adjacency matrix $A$ for graph $G=(V,E)$, number of vertices $n=|V|$, number of edges $m=|E|$, number of clusters $R \in \N$\;}
    $s_v \leftarrow NA$ for all $v \in V$\;
    $s_v^\ast \leftarrow 1$ for all $v \in V$\;
    Compute matrix $B_{vw}=A_{vw}-\frac{k_v k_w}{2m}$, where $k_v \in \N_0$ is the node degree for each $v \in V$\;
    $p \leftarrow 1$\;
    $b^\ast \leftarrow$ Algorithm~\ref{algo:explore}$(B,n,m,p,s,s^\ast,R)$\;
    \Return{$b^\ast$, $Q(b^\ast)$}\;
\end{algorithm}

The runtime of Algorithm~\ref{algo:branch} is still exponential in $|V|$ since in the worst case, all $R^{|V|}$ partial solutions have to be explored. Additionally, each computation of the bounds with Algorithm~\ref{algo:bounds} costs $O(|V|^2)$, resulting in a combined worst-case runtime for Algorithm~\ref{algo:branch} of $O(|V|^2 R^{|V|})$. Although this runtime does not seem like an improvement over the brute-force approach (see Section~\ref{sec:idea}), in practice Algorithm~\ref{algo:branch} allows one to considerably prune the search space and accelerate the exploration of all possible assignments of vertices to clusters (see Section~\ref{sec:results}).

\section{Experimental results}
\label{sec:results}
In this section, we assess Algorithm~\ref{algo:branch} with respect to the efficiency of the modularity bounds as well as its accuracy and runtime. We start with an outline of the experimental setting (Section~\ref{sec:setting}) and a brief review of the benchmark algorithms we consider (Section~\ref{sec:benchmarks}). Afterwards, we look at the efficiency of the bounding step in Section~\ref{sec:efficiency}, in particular at the proportion of partial or complete assignments of $s_v$ for $v \in V$ which are being explored by Algorithm~\ref{algo:branch}. In Section~\ref{sec:accuracy_runtime}, we compare Algorithm~\ref{algo:branch} to two benchmark algorithms.

\subsection{Setting}
\label{sec:setting}
We apply Algorithm~\ref{algo:branch} to random cluster graphs generated with the \textit{networkx} package in Python (using the function \textit{random\_partition\_graph}). This function requires three parameters, the number of nodes per cluster and the edge probabilities within and between clusters. The number of nodes per cluster is varied for all experiments and given in the figures. The edge probabilities are set to $p_{in} = 0.7$ (within clusters) and $p_{out} = 0.08$ (between clusters), respectively. We generate cluster graphs with $R=2$ and $R=3$ clusters in our experiments. All generated instances remain fixed, and we report results as averages across $50$ independent runs of our algorithms on each instance.

Our focus is on recovering the optimal value of the modularity of eq.~\eqref{eq:modularity}, that is the best value found via brute-forcing. We report the modularity-value gap, that is the gap between the modularity value returned by any method and the one obtained by brute-forcing eq.~\eqref{eq:modularity}. The modularity-value gap is always non-negative. Since the cluster labels are interchangeable while leaving the value of the modularity invariant, without loss of generality we fix $s_1=1$ (the first vertex is allocated to the first cluster) in both the brute-force solver and Algorithm~\ref{algo:branch}.

All computations were carried out on a MacBook Air with an Apple M2 chip (8 cores with 3.49 GHz each) and 8GB of RAM.

\subsection{Benchmark algorithms}
\label{sec:benchmarks}
We compare Algorithm~\ref{algo:branch} to two benchmark algorithms. Those are the Louvain method of \cite{Blondel2008}, and the simulated annealing algorithm of \cite{Kirkpatrick1983}. The Louvain method is a widespread greedy algorithm to detect clusters or communities using a multilevel approach. It aims to find small communities with a local search heuristic that optimizes the modularity value over all nodes. Afterwards, small communities are collapsed into super-nodes before repeating the first step. The Louvain algorithm was run with the help of the \textit{community} function of the \textit{networkx.algorithms} package in Python. Since the Louvain method automatically selects the number of communities, we ran it without a constraint on the number of communities and report the number of communities it detected in the appendix.

Simulated annealing (SA) is a multi-purpose optimization algorithm which is often used for discrete optimization tasks (such as the one of optimizing eq.~\eqref{eq:modularity}). Briefly, SA starts at a randomly initialized point in the search space and draws a random walk proposal from some distribution, either independently from the current state (a so-called independent proposal), or by (slightly) modifying the current state. The function value at the proposal is evaluated and a so-called acceptance probability for the proposal is computed. The acceptance probability has the property that (a) it is non-zero even if the proposal is inferior in quality to the current state, and (b) it decreases to zero over time (with the help of a so-called temperature function), thus slowly transitioning SA from an exploratory phase to a greedy optimization phase. If the proposal is accepted, the current state is overwritten and the algorithm continues with a new proposal. The algorithm usually stops after a pre-defined number steps, although several other termination criteria are possible. Our implementation of SA is provided on GitHub \citep{github}.

Moreover, in Section~\ref{sec:efficiency}, we also compare our algorithm to a brute-force solver that enumerates all $|V|^R$ assignments of $S_{vr}$ or $s_v$ in eq.~\eqref{eq:bounds}, respectively. The brute-force solver is guaranteed to globally maximize eq.~\eqref{eq:modularity} or eq.~\eqref{eq:bounds}, respectively.

\subsection{Efficiency of the bounding step}
\label{sec:efficiency}
We are interested in evaluating how well the bounds of Section~\ref{sec:bounds} restrict the search space and allow Algorithm~\ref{algo:branch} to discard partial solutions that cannot yield the global maximum. To this end, we apply both Algorithm~\ref{algo:branch} and the brute-force solver to cluster graphs of varying size using the setting of Section~\ref{sec:setting}. We measure the number of bound evaluations within Algorithm~\ref{algo:branch}, the number of recursive calls, and the number of complete assignments (meaning $s_v \neq NA$ for all $v \in V$).

\begin{figure*}
    \centering
    \includegraphics[width=0.49\textwidth]{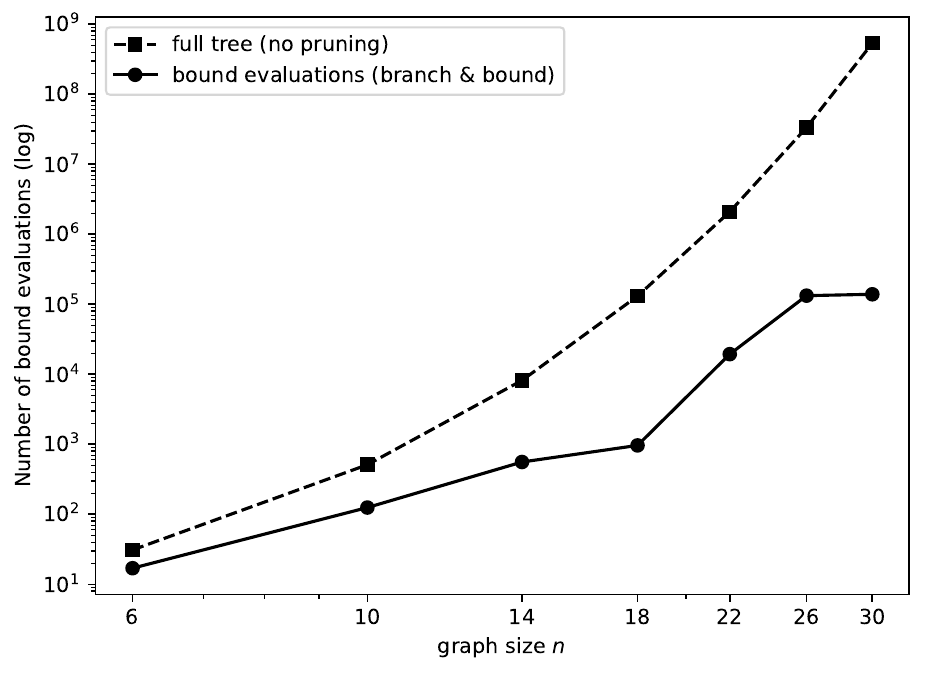}\hfill
    \includegraphics[width=0.49\textwidth]{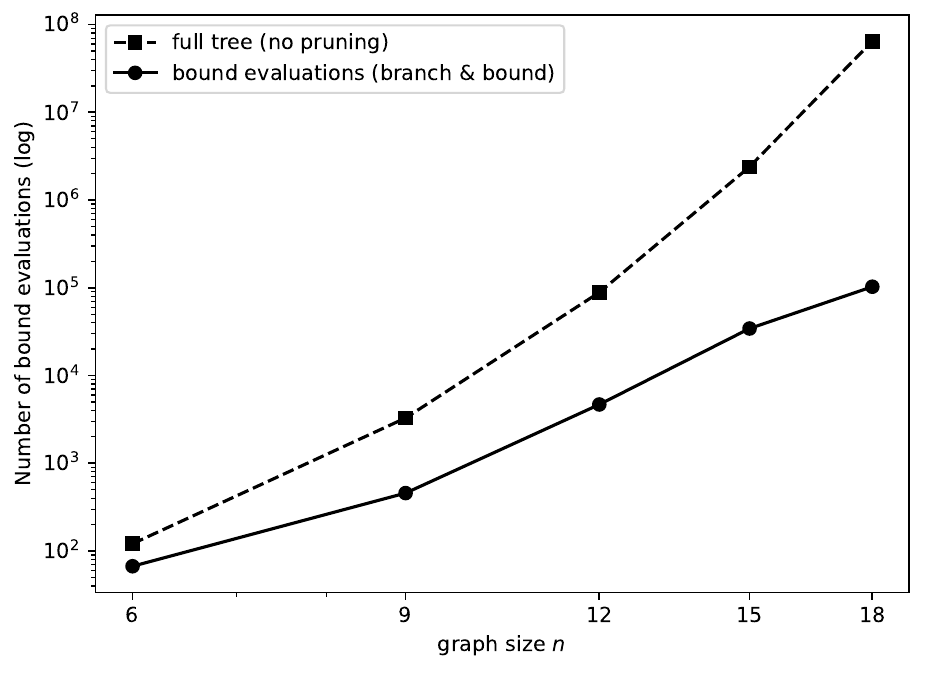}\\
    \includegraphics[width=0.49\textwidth]{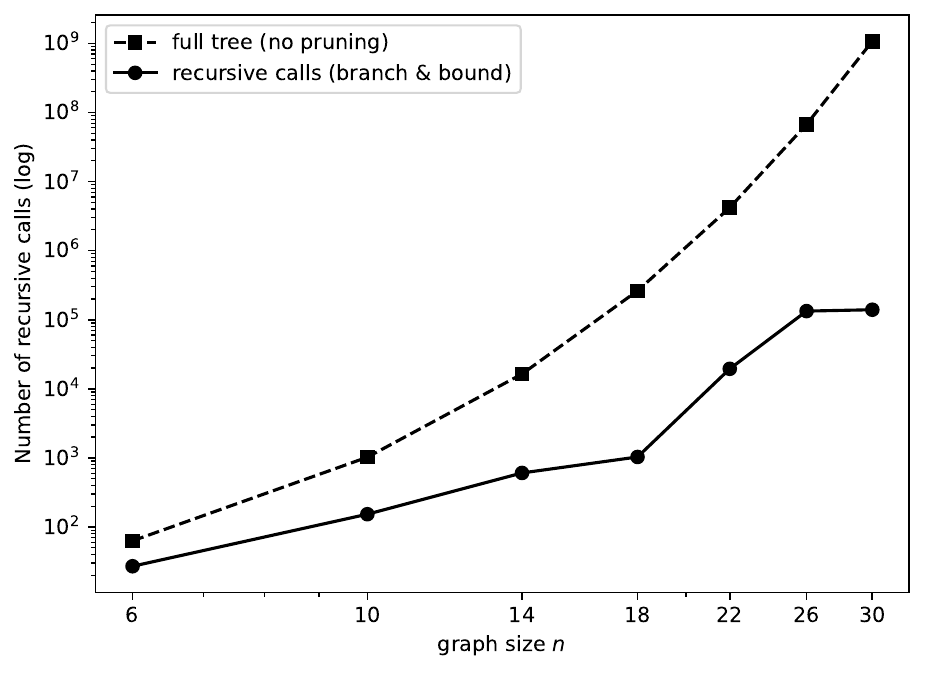}\hfill
    \includegraphics[width=0.49\textwidth]{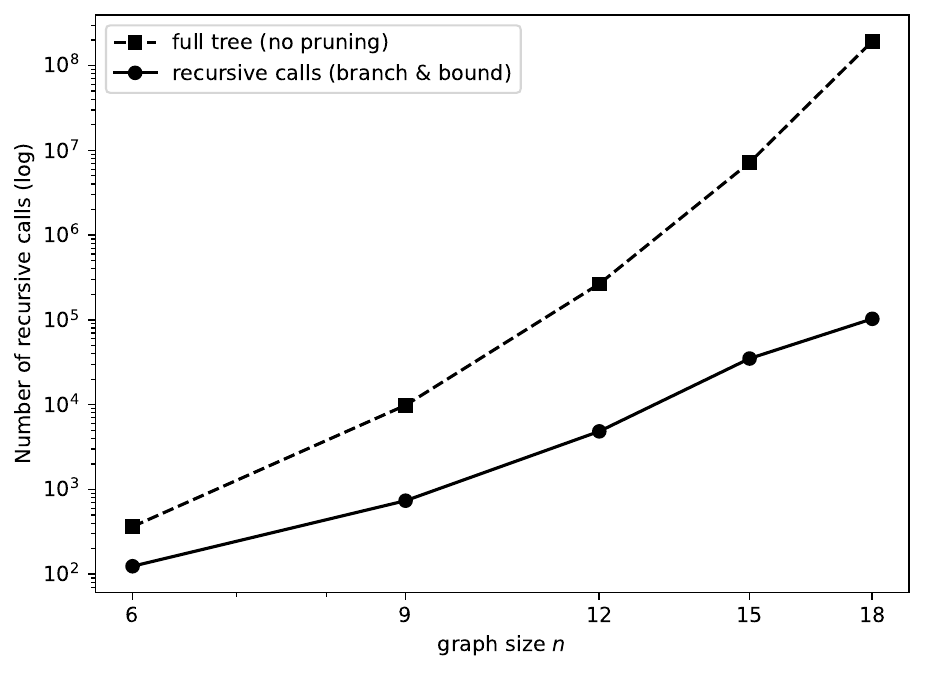}\\
    \includegraphics[width=0.49\textwidth]{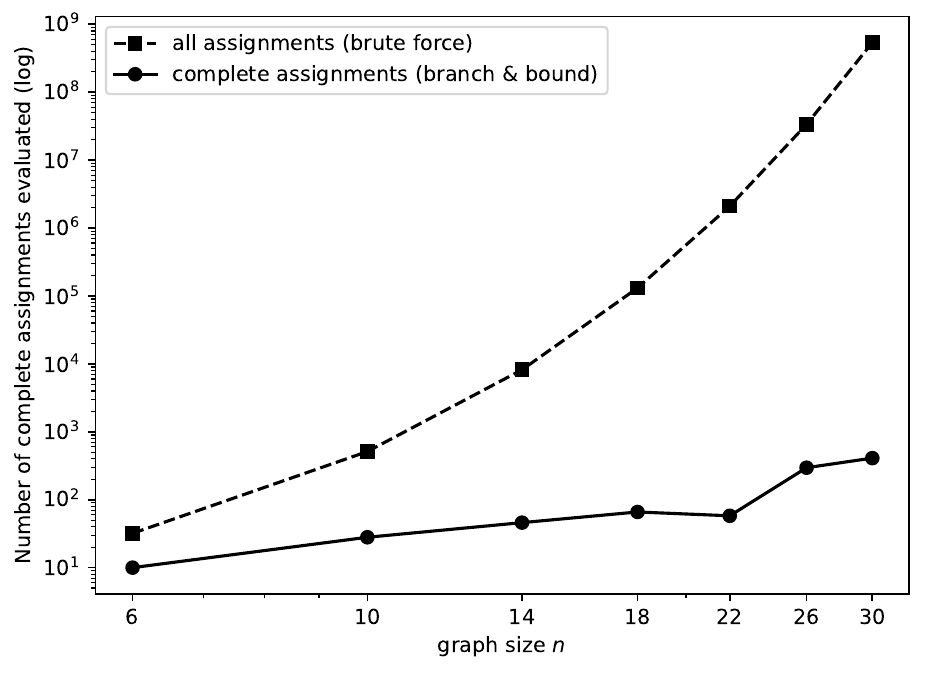}\hfill
    \includegraphics[width=0.49\textwidth]{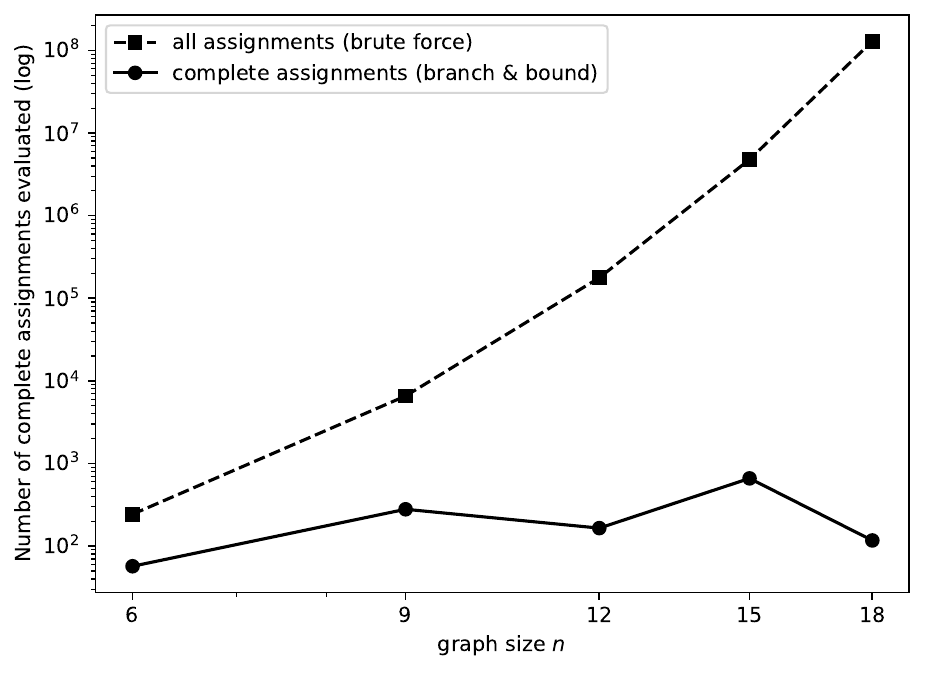}\\
    \caption{Application of Algorithm~\ref{algo:branch} and the brute-force solver to random cluster graphs of Section~\ref{sec:setting}. Number of bound evaluations (top), number of recursive calls (middle), and number of complete assignments reached (bottom) for Algorithm~\ref{algo:branch} as a function of the graph size $n$. The data points for the brute-force solver show the number of full assignments that were evaluated. Number of clusters $R=2$ (left) and $R=3$ (right). Both axes on the log scale in all plots.}
    \label{fig:enum}
\end{figure*}

Results are given in Figure~\ref{fig:enum} for $R=2$ and $R=3$ clusters. As a reference, Figure~\ref{fig:enum} also shows the number of all nodes in the full search tree visited by the brute-force solver (that is, all possible assignments of $s_v \in \{1,\ldots,R\}$ for $v \in V$). We observe that Algorithm~\ref{algo:branch} issues considerably less bound evaluations than there are nodes in the complete search tree of all possible assignments of $s_v$, $v \in V$. The number of bound evaluations is equivalent to the number of visited search-tree nodes, since the bounds are evaluated at all visited nodes. The number of recursive calls is proportional to the number of bound evaluations, as expected. Due to the pruning step, we observe that the number of complete assignments (meaning $s_v \neq NA$ for all $v \in V$) reached by Algorithm~\ref{algo:branch} is considerably smaller than the number explored by the brute-force solver.

\begin{figure*}
    \centering
    \includegraphics[width=0.49\textwidth]{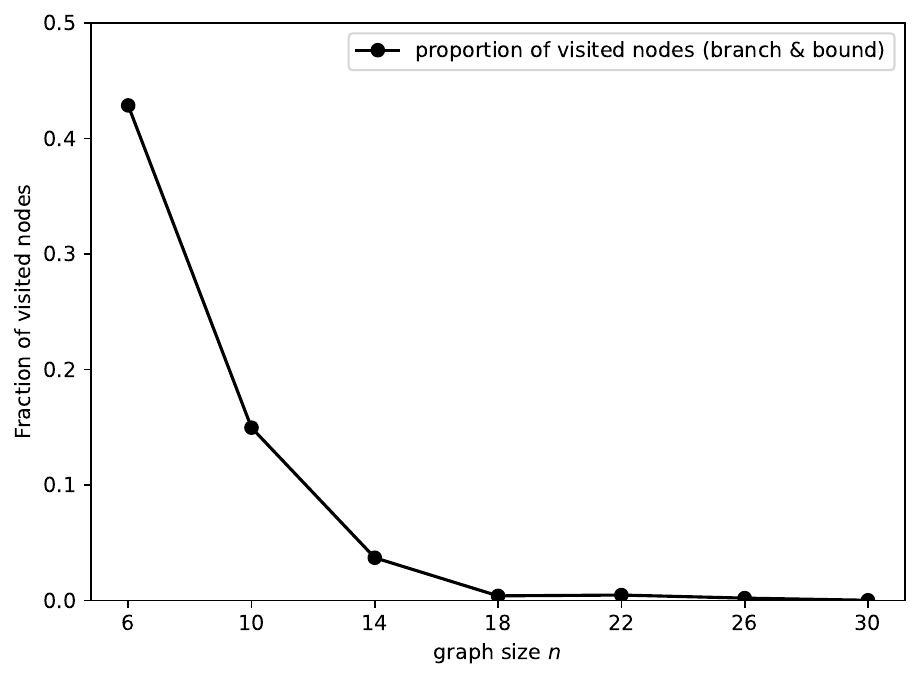}\hfill
    \includegraphics[width=0.49\textwidth]{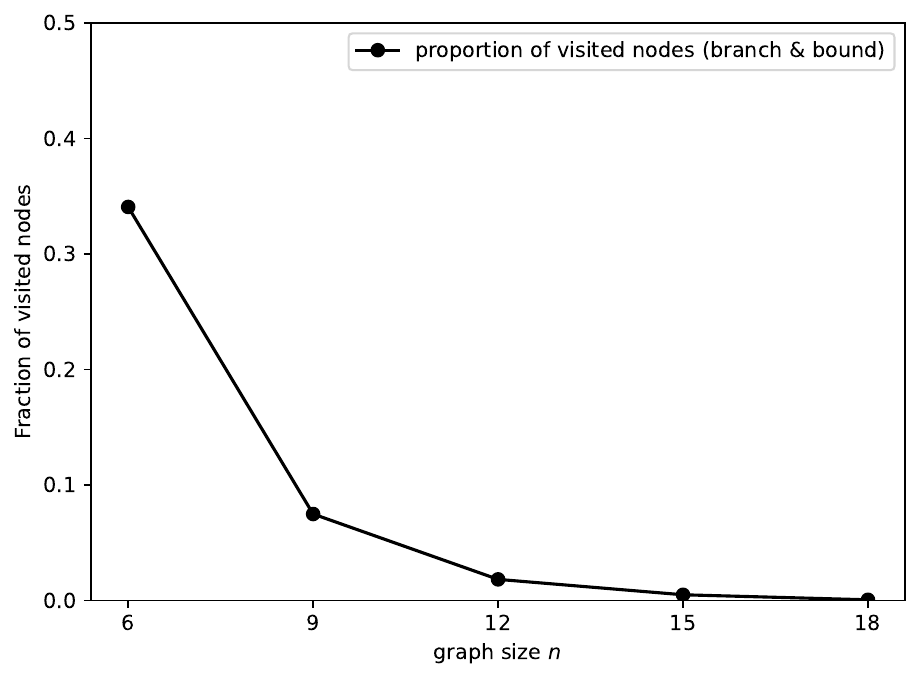}\\
    \caption{Application of Algorithm~\ref{algo:branch} to random cluster graphs of Section~\ref{sec:setting}. Proportion of the full search tree which is visited by Algorithm~\ref{algo:branch} for the results reported in Figure~\ref{fig:enum}. Number of clusters $R=2$ (left) and $R=3$ (right).}
    \label{fig:visited}
\end{figure*}

Figure~\ref{fig:visited} highlights the proportion of the full search tree which is visited by Algorithm~\ref{algo:branch}. The displayed data points correspond to the experiment reported in Figure~\ref{fig:enum}. We observe that the pruning works better the larger the graph is, showing that the proportion of the full search tree which is actually explored by Algorithm~\ref{algo:branch} decreases with the graph size $n$.

\subsection{Accuracy and runtime of the proposed algorithm}
\label{sec:accuracy_runtime}
We assess both accuracy and runtime of Algorithm~\ref{algo:branch} in comparison with the Louvain method and SA (Section~\ref{sec:benchmarks}). For this we again consider the experimental setting of Section~\ref{sec:setting}, meaning random cluster graphs with $R=2$ or $R=3$ clusters and varying size given in the figures.

We report the modularity-value gap, defined as the difference between the modularity value obtained by enumeration of all assignments in eq.~\eqref{eq:modularity} (as done by the brute-force solver) and the maximal modularity value returned by any method. If the modularity-value gap is zero for any algorithm, then the respective algorithm was able to return the optimal modularity value that can be achieved. Positive values denote a discrepancy from the best possible value.

\begin{figure*}
    \centering
    \includegraphics[width=0.49\textwidth]{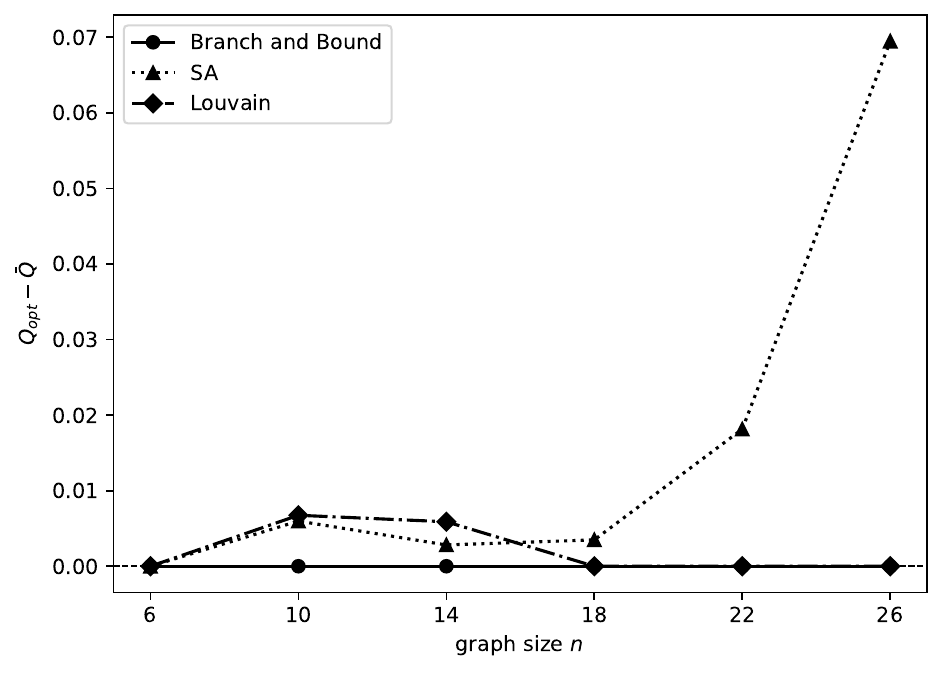}\hfill
    \includegraphics[width=0.49\textwidth]{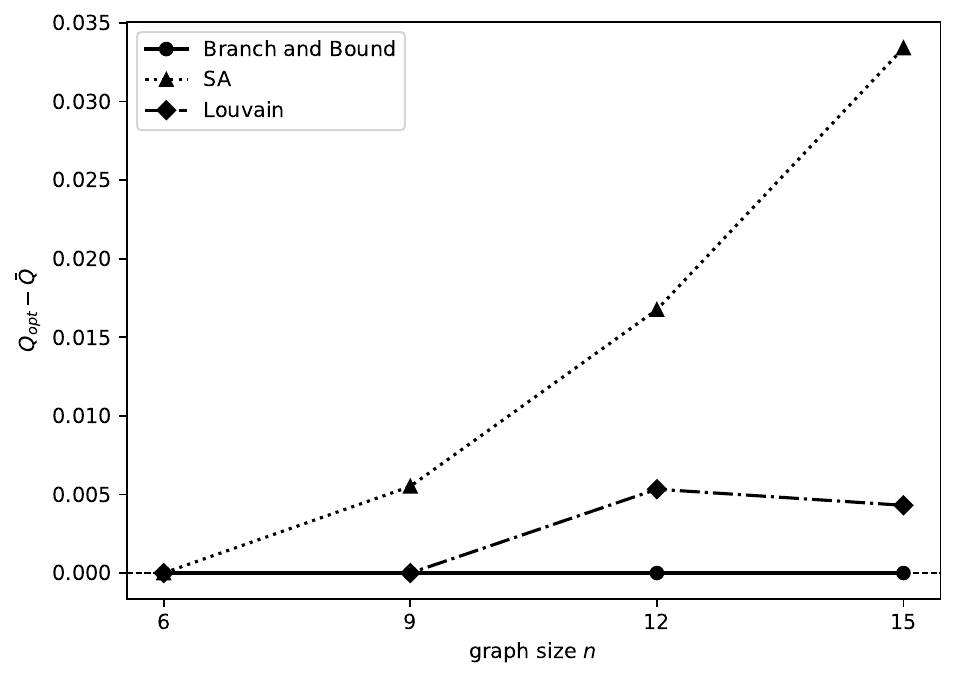}
    \caption{Application of Algorithm~\ref{algo:branch}, SA, and the Louvain method to random cluster graphs of Section~\ref{sec:setting}. Modularity-value gap (defined as the difference in modularity value between each algorithm and the value obtained with the brute-force solver) as a function of the graph size $n$. Number of clusters $R=2$ (left) and $R=3$ (right).}
    \label{fig:accuracy}
\end{figure*}

Figure~\ref{fig:accuracy} shows our results. We observe that for $R=2$ and $R=3$ clusters, our branch-and-bound approach of Algorithm~\ref{algo:branch} finds the optimal modularity value in all scenarios, meaning the best possible value that can be achieved in eq.~\eqref{eq:modularity} with brute-forcing. This is expected, see Theorem~\ref{thm:correctness}. For $R=2$, the Louvain method shows a small modularity-value gap for small $n$ but finds the optimal value for larger $n$ under consideration. However, for $R=3$ we already observe consistent discrepancies to the optimal modularity value for the Louvain method as $n$ increases. SA seems to be less suited for the task of modularity maximization and shows considerable discrepancies for large $n$. The number of communities detected by the Louvain method is given in Figure~\ref{fig:louvain} in Appendix~\ref{sec:appendix}. Figure~\ref{fig:louvain} shows that number of communities detected by the Louvain method mostly equals the one we preset when running both Algorithm~\ref{algo:branch} and SA, apart from the case of very small graph sizes.

When looking at the runtime, it is worth noting that both accuracy and runtime should not be considered separately. Indeed, fast runtimes are easy to achieve for the task of maximizing eq.~\eqref{eq:modularity} at the expense of a deterioration in quality. For this reason, we report both the wall-clock runtime in seconds and the Time-To-Solution (TTS) metric \citep{maxclique_qtop}. The TTS metric incorporates both the runtime and a measure of the solution quality, precisely the probability of finding the optimal (brute-forced) modularity value of eq.~\eqref{eq:modularity}. The TTS metric is defined as
\begin{align}
    \text{TTS} = T \cdot \frac{\log(1-0.99)}{\log(1-p)},
    \label{eq:tts}
\end{align}
where $T$ is the mean runtime among a fixed number of $N \in \N$ replicates (in our case, $N=50$, see Section~\ref{sec:setting}), and $p$ is the empirical probability of finding the optimal (brute-forced) modularity value among those $N$ runs. The TTS metric reflects the time it takes to reach an optimum solution with a $99$ percent confidence. If $p=1$ in eq.~\eqref{eq:tts}, the optimal solution was always found and we report the average (deterministic) runtime instead of the TTS. If $p=0$, the optimal solution was never found, and we define TTS as infinity. This case did not occur in our experiments.

\begin{figure*}
    \centering
    \includegraphics[width=0.49\textwidth]{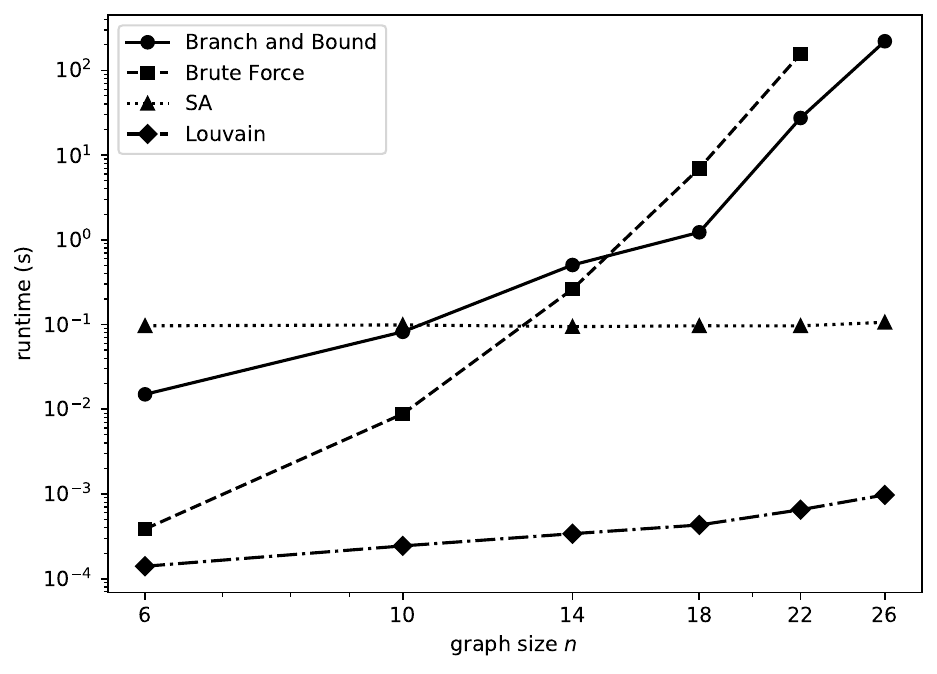}\hfill
    \includegraphics[width=0.49\textwidth]{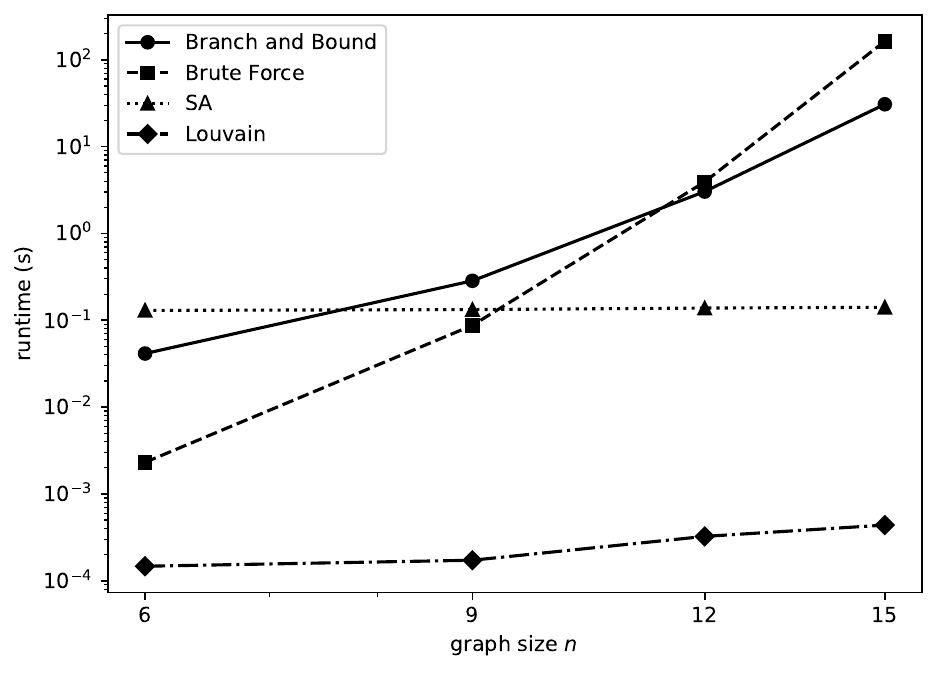}\\
    \includegraphics[width=0.49\textwidth]{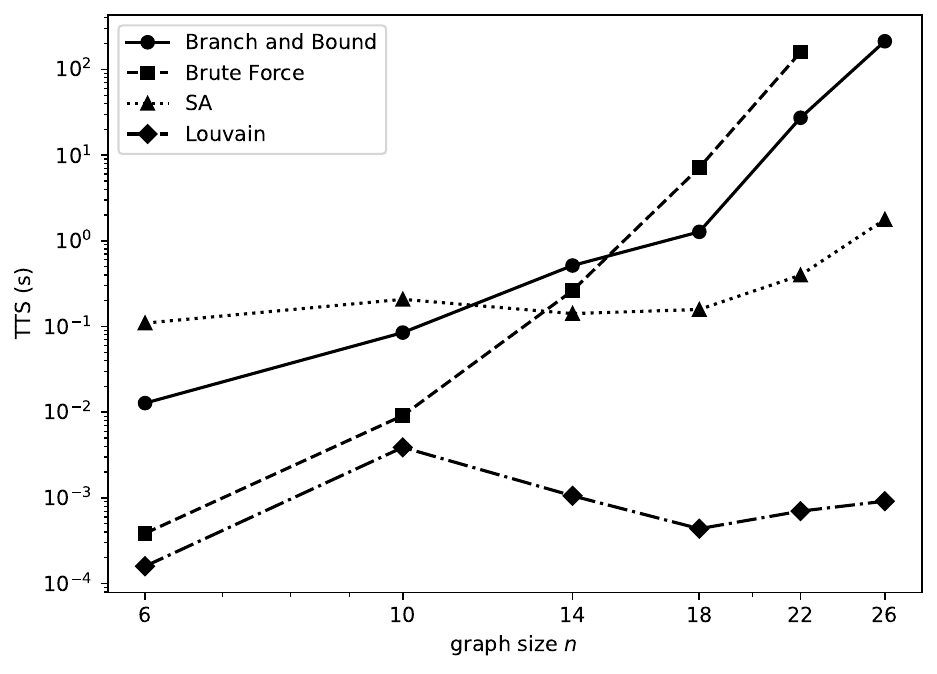}\hfill
    \includegraphics[width=0.49\textwidth]{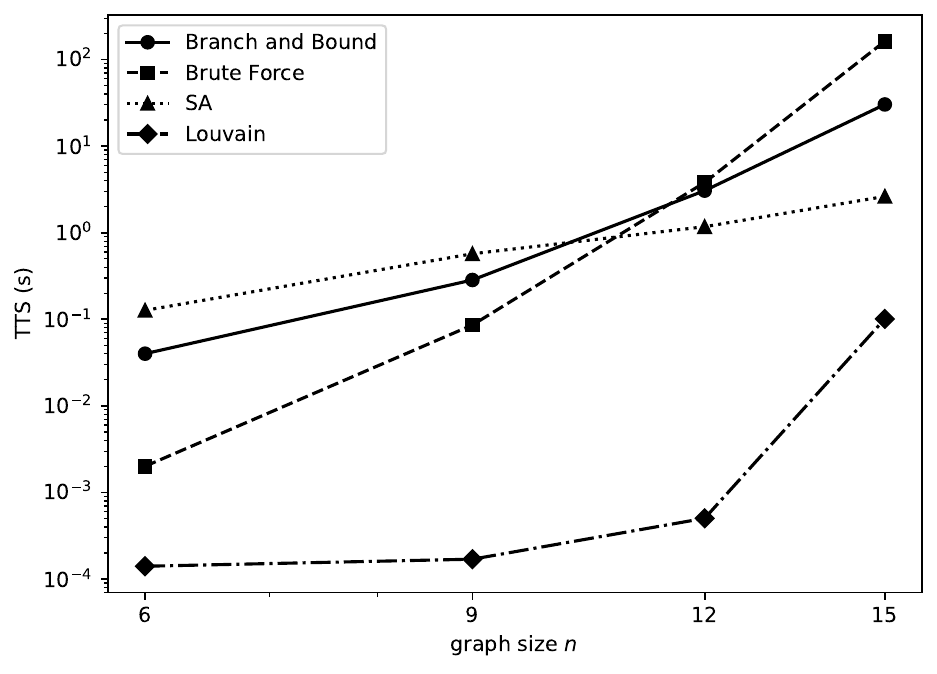}
    \caption{Application of Algorithm~\ref{algo:branch}, SA, and the Louvain method to random cluster graphs of Section~\ref{sec:setting}. Wall-clock runtime (top) and TTS metric (bottom) in seconds for the results reported in Figure~\ref{fig:accuracy}. Number of clusters $R=2$ (left) and $R=3$ (right). Both axes on the log scale in all plots.}
    \label{fig:runtime}
\end{figure*}

Figure~\ref{fig:runtime} shows the wall-clock runtimes and TTS times for Algorithm~\ref{algo:branch}, the brute-force solver, the Louvain method, and SA. Looking at the wall-clock runtimes for both $R=2$ and $R=3$, we observe that the Louvain method achieves the fastest execution times. We attribute this to the optimized implementation in the \textit{networkx} package in Python. The runtimes of SA are constant, which is due to the fact that we run SA with a cutoff on the number of iterations as one of the termination critertia. This is standard practice for the SA algorithm. Importantly, we observe that Algorithm~\ref{algo:branch} is slower than the brute-force solver for small graph sizes, which is expected due to the overhead stemming from the bound computations with Algorithm~\ref{algo:bounds}. For larger $n$, Algorithm~\ref{algo:branch} achieves considerably lower runtimes than the brute-force approach.

This picture is confirmed when looking at the TTS metric that takes into account the accuracy of all methods. Here, the Louvain method is the fastest method for $R=2$ and $R=3$, although in the case of $R=3$ we observe a steep increase in TTS as the graph size increases. The TTS metric for SA is not constant anymore, demonstrating the fact that the quality of the SA solutions deteriorates with increasing graph size. As observed for wall-clock runtimes, Algorithm~\ref{algo:branch} is slower than the brute-force solver for small graph sizes and achieves a more favorable runtime scaling than the brute-force solver as the graph size increases.

\section{Discussion}
\label{sec:discussion}
In this article, we propose a branch-and-bound algorithm that maximizes the network modularity, a widespread metric that measures the division of a network into modules or clusters. Our algorithm is proven to globally maximize the modularity value, however, it achieves this without a full exploration of the search space of all $R^{|V|}$ possible assignments, where $V$ is the set of vertices and $R$ is the number of clusters. This is possible with the help of guaranteed lower and upper bounds on the network modularity for incomplete assignments of vertices to clusters.

In a simulation study, we demonstrate that the bounds allow for considerable pruning of the search space in comparison to a brute-force approach that enumerates all possible solutions. This allows us to find optimal solutions with our branch-and-bound algorithm while reducing runtimes in practice.

The worst-case runtime of our branch-and-bound algorithm is $O(|V|^2 R^{|V|})$. This is higher than the complexity of $O(R^{|V|})$ for the brute-force approach, due to a polynomial overhead for the computation of the bounds. This seemingly contradicts the experimental results presented in this article. However, despite having a higher worst-case complexity, we demonstrate that the computation of the bounds yields a considerable reduction in the number of explored search-tree nodes in practice.

Our algorithm also has several limitations. First, its worst-case runtime complexity is still exponential. Nevertheless, this is to be expected for an exact algorithm, since the task of globally maximizing the network modularity is NP-hard. Second, our algorithm assumes that the (maximal) number of clusters $R$ is given, while other algorithms, such as the Louvain algorithm, detect it at runtime. This can be an advantage or disadvantage depending on the application. Third, our algorithm does not impose that all clusters be used, meaning it is not guaranteed that there is at least one node per cluster. Indeed, if using fewer than $R$ clusters leads to a larger value of the network modularity given in eq.~\eqref{eq:modularity}, the algorithm might use fewer clusters. Fourth, due to our focus on the exact optimization of the modularity value, our simulations are restricted to small graph sizes. This issue is tied to a potential looseness of the bound we derive on the modularity value. We hope to address this issue in future work by concentrating on the development of sharper bounds on the modularity value. Another research direction could aim to tighten the bounds by making them probabilistic confidence bands. This would likely result in stronger pruning, however, the usage of probabilistic confidence bands comes at the expense of losing the optimality property established for our branch-and-bound algorithm. Fifth, our algorithm does not consider label symmetry or order dependence apart from the fact that we fix $s_1=1$. Since in our current implementation, all permutations of the node labels are explored independently, our algorithm has some degree of avoidable redundancy. More sophisticated symmetry-breaking techniques would help to narrow down the search space and improve the runtime further. Sixth, our experiments are based on simulated data and do not consider real-world networks. However, due to the limitation to small graph sizes, we leave the validation on real-world networks for future work, ideally when more sophisticated bounds enabling stronger pruning are available.

\appendix
\section{Additional simulations}
\label{sec:appendix}
The appendix provides two more panels on the Louvain method, see Figure~\ref{fig:louvain}. Since the implementation of the Louvain method available in \textit{networkx} detects the number of communities $R$ by itself, we provide the number of communities found by the Louvain method for the results reported in Figure~\ref{fig:accuracy}. Each data point in Figure~\ref{fig:louvain} shows averages across the $50$ runs with error bars of one standard deviation. The figure demonstrates that the number of communities detected corresponds to the target number we use when calling Algorithm~\ref{algo:branch} for large graph sizes $n$. For small $n$, especially in the case of $R=2$, the Louvain method might detect more communities than anticipated. We attribute this to the small graph size $n$, which causes greater volatility in the detection of the community number.

\begin{figure*}
    \centering
    \includegraphics[width=0.49\textwidth]{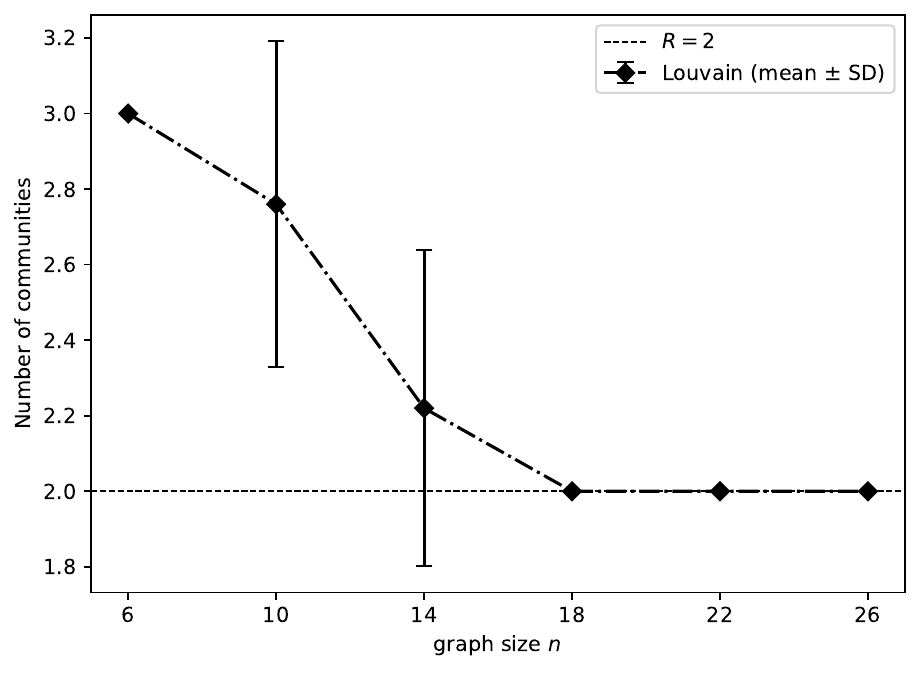}\hfill
    \includegraphics[width=0.49\textwidth]{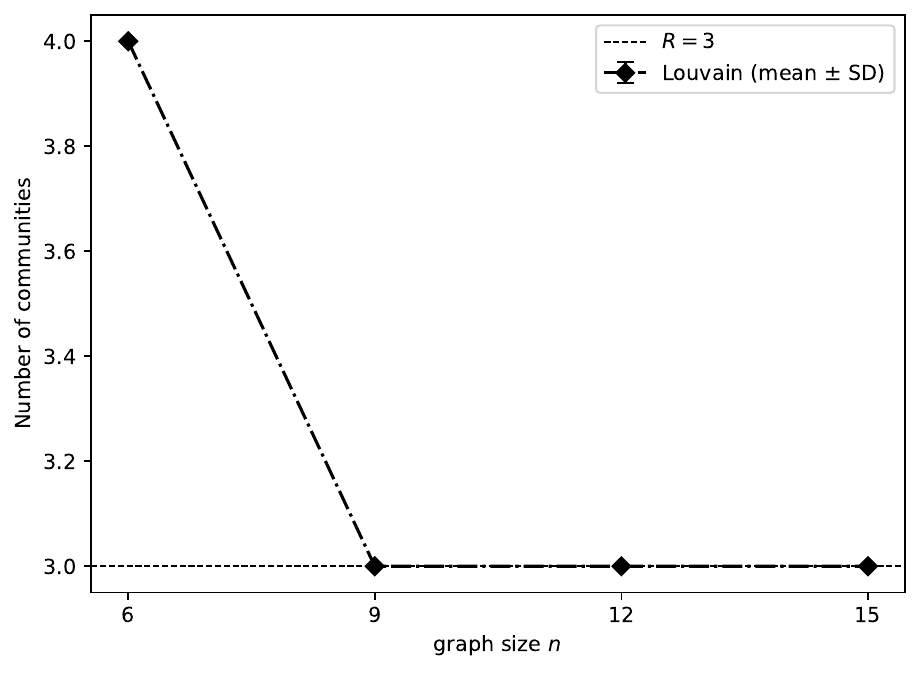}
    \caption{Number of communities found by the Louvain method for the results reported in Figure~\ref{fig:accuracy}. Data points show averages across the $50$ runs with error bars of one standard deviation. Target number of clusters $R=2$ (left) and $R=3$ (right) indicated with horizontal lines.}
    \label{fig:louvain}
\end{figure*}

\bibliographystyle{apalike}
\bibliography{main}
\end{document}